\documentclass{article}
\usepackage{graphicx}
\usepackage{amsmath}
\usepackage{amssymb}
\usepackage{amsthm}
\usepackage{mathrsfs}
\newtheorem{dfn}{Definition}[section]
\newtheorem{thm}[dfn]{Theorem}
\newtheorem{prop}[dfn]{Proposition}
\newtheorem{lem}[dfn]{Lemma}
\newtheorem{cor}[dfn]{Corollary}
\newtheorem{rem}[dfn]{Remark}

\newtheorem{ass}[dfn]{Assumption}
\newtheorem{ex}[dfn]{Example}

\numberwithin{equation}{section}

\begin{document}

\title{A note on the spectral mapping theorem\\ of quantum walk models}

\author{Kaname Matsue\thanks{The Institute of Statistical Mathematics, Tachikawa, Tokyo, 190-8562, Japan} $^{,}$\footnote{\tt kmatsue@ism.ac.jp} ,\ Osamu Ogurisu\thanks{Division of Mathematical and Physical Sciences,
  Kanazawa University, Kanazawa, Ishikawa 920-1192, Japan} $^{,}$\footnote{\tt ogurisu@staff.kanazawa-u.ac.jp}$\ $ and Etsuo Segawa\thanks{Graduate School of Information Sciences, Tohoku University, Aoba, Sendai, 980-8579, Japan}  $^{,}$\footnote{\tt e-segawa@e.tohoku.ac.jp}}
\maketitle
\begin{abstract}
We discuss the description of eigenspace of a quantum walk model $U$ with an associating linear operator $T$ in abstract settings of quantum walk including the Szegedy walk on graphs.
In particular, we provide the spectral mapping theorem of $U$ without the spectral decomposition of $T$.
Arguments in this direction reveal the eigenspaces of $U$ characterized by the generalized kernels of linear operators given by $T$.
\end{abstract}


\section{Introduction}
Quantum walks are quantum analogues of classical random walks.
Their primitive forms of the discrete-time quantum walks on $\mathbb{Z}$ can be seen in Feynman's checker board \cite{FH}. 
It is mathematically shown (e.g. \cite{K1}) that this quantum walk has a completely different limiting behavior from classical random walks, which is a typical example showing a  difficulty of intuitive description of quantum walks' behavior. 

One of main aims of studies of quantum walks from the mathematical point of view is to understand their asymptotic behavior.
There are two typical approaches for detecting asymptotic behavior of quantum walks:
\begin{itemize}
\item Calculation of density functions for long time limits of quantum walks;
\item Description of the spectrum of quantum walks as unitary operators.
\end{itemize}

In \cite{HKSS}, the spectral mapping theorem of the twisted Szegedy walk $U$ is derived with spectral decomposition of the associated self-adjoint operator $T$.
According to \cite{HKSS}, the eigenstructure of $T$ induces those of the operator of the form
\begin{equation*}
\tilde T = \begin{pmatrix}
0 & -I\\
I & 2T
\end{pmatrix},
\end{equation*}
where $I$ is the identity on an appropriate linear space, and eigenstructure of $\tilde T$ determines an invariant subspace of $U$.
As mentioned before, their arguments rely on the spectral decomposition and the eigenstructures of $T$.

\bigskip
In this paper, we propose a spectral analysis method of $U$ {\em without directly using the spectral decomposition of $T$}.
The motivation of this study is to overcome the difficulty concerning with spectral structure of $T$ such as a quantum walk model discussed in \cite{MOS}.
As a first step in this direction, we try to apply our new method to the problems whose spectral structures have been well developed, that is, 
Szegedy walks \cite{P2016, Sze} and its abstract quantum walks \cite{O2016, SS2015-1, SS2015-2}. 
We obtain a new observation of $U$ by this method which has not discussed well before, which is given as follows. 
Let
\begin{equation}
\label{gspec}
{\rm Spec}(A) = \left\{\lambda \in \mathbb{C}\mid 0\not = \exists \psi \in \bigcup_{n\in \mathbb{N}} \ker(\lambda I -A)^n\right\}
\end{equation}
for a linear operator $A$ on a Hilbert space. Then
\begin{itemize}
\item As for $\lambda \in {\rm Spec}(U|_{\mathcal{L}})\setminus \{\pm 1\}$, we have $\ker (\lambda I -U|_{\mathcal{L}}) = L(\ker(\lambda I - \tilde T))$.
\item As for $\lambda  \in {\rm Spec}(U|_{\mathcal{L}})\cap \{\pm 1\}$, we have $\ker (\lambda I -U|_{\mathcal{L}}) = L(\ker(\lambda I - \tilde T)^2\setminus \ker(\lambda I - \tilde T))$. 
\end{itemize}
Detailed descriptions of $U$, $L$ and $\mathcal{L}$ are shown in Sections \ref{section-absQW} and \ref{section-spec}.
The new insight of our study is the presence of the generalized eigenspace of the linear operator $\tilde T$.
We expect that such generalized eigenstructures reflect not only the geometric feature of underlying graphs such as their bipartiteness and underlying random walks such as their reversibility (\cite{HKSS}), but also performance of quantum search algorithms on graphs \cite{P2016, Sze}.
We also expect that our result explicitly reveals such hidden structure and will lead to deeper study of spectra and asymptotic behavior of quantum walks from the viewpoint of functional analysis and geometry.

Throughout our discussions, we consider an abstract {\em quantum walk model} given below, which extracts the essence of well-known Szegedy walks on graphs (e.g. \cite{HKSS}).
Our study will cover spectral analysis for a general class of quantum walks (e.g. \cite{MOS, P2016})
\par
Remark that there are preceding works of quantum walks in such an abstract setting: \cite{SS2015-1, SS2015-2}.
There quantum walks on {\em infinite dimensional} Hilbert spaces are considered.
On the other hand, we restrict our considerations to finite dimensional spaces in this paper.

\section{Abstract quantum walk models}
\label{section-absQW}
Throughout this paper, we study the spectrum of quantum walks in the following setting.
Note that the following settings are finite dimensional analogue of \cite{SS2015-1, SS2015-2}.
\begin{itemize}
\item $K_1$ and $K_2$ : finite dimensional Hilbert spaces over $\mathbb{C}$ with inner products $\langle \cdot, \cdot \rangle_{K_i}$.
\item $S: K_2\to K_2$ : a self-adjoint, unitary operator. 
\item $d_A : K_2\to K_1$ : a bounded linear operator with the adjoint operator $d_A^\ast : K_1\to K_2$.
\item $d_B : K_2\to K_1$ : a bounded linear operator given by $d_B = d_A S$.
The adjoint operator $d_B^\ast$ is given by the similar way to $d_A^\ast$.
\item $T : K_1\to K_1$ : a bounded linear operator given by $T = d_A d_B^\ast$, which is called {\em the discriminant operator}.
\end{itemize}
Note that the linear operator $T$ is actually self-adjoint since $S:K_2\to K_2$ is self-adjoint and unitary.

Now we assume the following property, which is crucial to our setting.
\begin{ass}
\label{ass-identity}
$d_A d_A^\ast = I : K_1\to K_1$.
\end{ass}
\begin{lem}
\label{lem-coin}
Under Assumption \ref{ass-identity}, the linear operator $C := 2d_A^\ast d_A - I : K_2 \to K_2$ is a self-adjoint and unitary operator.
\end{lem}
\begin{proof}
Direct calculations yield
\begin{align*}
\langle (d_A^\ast d_A) f,g \rangle_{K_2} = \langle d_A f, d_A g \rangle_{K_1} = \langle f, (d_A^\ast d_A) g \rangle_{K_2}
\end{align*}
holds for all $f,g\in K_2$, which implies that $d_A^\ast d_A$, and hence $C$, is self-adjoint.
Therefore, it is sufficient to prove that $C^\dagger C = C^2 = I : K_2\to K_2$.
We immediately have
\begin{equation*}
C^2 = (2d_A^\ast d_A - I)^2 = 4d_A^\ast d_A d_A^\ast d_A - 4d_A^\ast d_A + I = 4d_A^\ast d_A - 4d_A^\ast d_A + I = I,
\end{equation*}
which shows the statement. 
Note that we have used Assumption \ref{ass-identity} in the above calculation.
\end{proof}

Our quantum walk model is given by the following definition.
\begin{dfn}[Quantum walk model]\rm
Let $C$ be the unitary operator given in Lemma \ref{lem-coin}.
Then the operator $U=SC : K_2 \to K_2$ is also a unitary operator.
We shall say the operator $U$ {\em a quantum walk model on $K_2$} associated with the pair $(K_1, d_A)$ of additional Hilbert space $K_1$ and the linear operator $d_A$ acting on it.
\end{dfn}
Note that the discriminant operator $T$, which is the center of our considerations, and the operator $d_B$ are naturally defined by $d_A$ and $S$.
\par
Now we have defined the unitary operator $U$ as a quantum walk, while $U$ may not be seen as a \lq\lq quantum walk" at a glance.
In fact, the operator $U$ is an abstract model of well-known quantum walks such as {\em Grover walk} and {\em Szegedy walk} as follows.

\begin{ex}[Szegedy walk on a graph]\rm
\label{ex-Szegedy}
Let $G=(V(G), E(G))$ be a simple and finite graph, where $V(G)$ is the set of vertices in $G$ and $E(G)$ is the set of (undirected) edges in $G$.
It can be regarded as the digraph $G=(V(G), D(G))$, where $D(G) = \{e, \bar e \mid e\in E(G)\}$ and $\bar e = (v,u)$ for each $e=(u,v)$, $u,v\in V(G)$.
For each edge $e=(u,v)\in D(G)$, $o(e) = u$ denotes the origin of $e$ and $t(e) = v$ denotes the terminal point of $e$.
\par
Now define a $\mathbb{C}$-linear space $\ell^2(D(G))$ by 
\begin{equation*}
\ell^2(D(G)):= \left\{f : D(G)\to \mathbb{C} \mid \| f \|_{D(G)} < \infty\right\}.
\end{equation*}
Here the inner product is given by the standard inner product, that is,
\begin{equation*}
\langle f, g \rangle_{D(G)}:= \sum_{e \in D(G)}\overline{f(e)}g(e).
\end{equation*}
Let $\|\cdot \|_{D(G)}$ be the associated norm, namely, $\|f\|_{D(G)} := \langle f, f \rangle_{D(G)}^{1/2}$.
We take
\begin{equation*}
\delta^{(1)}_{e}(e'):= \begin{cases}
	1 & \text{ if $e' = e$}\\	
	0 & \text{ if $e' \not = e$}
\end{cases}
\end{equation*}
as the standard basis of $\ell^2(D(G))$. One knows that the $\mathbb{C}$-linear space $\ell^2(D(G))$ associated with the inner product $\langle \cdot, \cdot \rangle_{D(G)}$ is a Hilbert space. 
We can also define the Hilbert space $\ell^2(V(G))$ in the similar manner.
\par
Next, call a function $w : D(G)\to \mathbb{C}$ {\em a weight} if $w(e)\not = 0$ for all $e\in D(G)$ and 
\begin{equation*}
\sum_{e:o(e)=u}|w(e)|^2 = 1\quad \text{ for all }u\in V(G).
\end{equation*}
Let $S : \ell^2(D(G)) \to \ell^2(D(G))$ be defined by $Sf(e) = f(\bar e)$, called the shift operator.
Under such settings, define $d_A, d_B : \ell^2(D(G))\to \ell^2(V(G))$ as
\begin{equation*}
(d_A \phi)(v) = \sum_{e:o(e)=v} \overline{w(e)}\phi(e),\quad (d_B \phi)(v) = \sum_{e:o(e)=v} \overline{w(e)}\phi(\bar e),
\end{equation*}
respectively. It immediately follows that $d_B = d_AS$. 
Their adjoints $d_A^\ast, d_B^\ast : \ell^2(V(G))\to \ell^2(D(G))$ are defined by 
\begin{equation*}
(d_A^\ast \psi)(e) = w(e)\psi(o(e)),\quad (d_B^\ast \psi)(e) = w(\bar e)\psi(t(e)),
\end{equation*}
from the relationship $\langle \phi, d_J^\ast \psi \rangle_{D(G)} = \langle d_J \phi, \psi \rangle_{V(G)}\ (J\in \{A,B\})$ for all $\psi \in \ell^2(V(G))$ and $\phi \in \ell^2(D(G))$.
Then, from the property of the weight $w$, we can prove that $d_A d_A^\ast = d_B d_B^\ast = I : \ell^2(V(G))\to \ell^2(V(G))$ (cf. \cite{HKSS}). 
In particular, the Szegedy walk $U=SC = S(2d_A^\ast d_A - I)$ in this setting is contained in our current setting.
We often call unitary operators $S$ {\em the shift operator} and $C$ {\em the quantum coin operator}.
The discriminant operator $T = d_A d_B^\ast$ is also defined in the natural way.
\par
If we further assume that $w(e) = 1/\sqrt{\deg(o(e))}$ for all $e\in E(G)$, the resulting quantum walk model $U$ is nothing but the Grover walk on $G$.
\end{ex}


\section{Spectral analysis of abstract quantum walk models}
\label{section-spec}
\subsection{Invariant subspaces of $U$}
Now we consider ${\rm Spec}(U)$, the spectrum of $U$ in the sense of (\ref{gspec}).
As seen in preceding works such as \cite{HKSS}, ${\rm Spec}(U)$ consists of eigenvalues inherited from those of a self-adjoint operator $T:K_1\to K_1$ via the spectral mapping property and specific ones to $U$.
\begin{rem}\rm
Since $U$ is a normal operator, we do not usually need to consider the spectrum in the sense of (\ref{gspec}).
However, in the consideration of the {\em eigensystem} of $U$, we need the notion of generalized eigenspaces.
This is why we introduce (\ref{gspec}).
\end{rem}

To characterize the spectral mapping property of ${\rm Spec}(U)$, we consider the following operators.
Let $L: K_1^2\to K_2$ by $L(f,g)^T =d_A^\ast f+d_B^\ast g$, where $d_B^\ast =Sd_A^\ast$.
Note that $L$ is a linear map.
Indeed, for any $\alpha_i, \in \mathbb{C}$ and $\psi_i =(f_i,g_i)^T\in K_1^2$, $i=1,2$, we have
\begin{align*}
L(\alpha_1 \psi_1 + \alpha_2 \psi_2) &= L(\alpha_1 f_1+\alpha_2 f_2, \alpha_1 g_1+\alpha_2 g_2)^T\\
	&=d_A^\ast (\alpha_1 f_1+\alpha_2 f_2) + d_B^\ast (\alpha_1 g_1+\alpha_2 g_2)\\
	&=\sum_{i=1}^2 \alpha_i(d_A^\ast f_i + d_B^\ast g_i) = \sum_{i=1}^2 \alpha_i L\psi_i.
\end{align*}
\par
Also, let  
$\tilde{T}:K_1^2\to K_1^2$ by 
\begin{equation}
\label{tilde-T}
\tilde{T}=\begin{pmatrix} 0 & -I \\ I & 2T \end{pmatrix}.
\end{equation}

\begin{lem}
\label{ULLT}
\begin{equation} UL=L\tilde{T} : K_1^2 \to K_2.
\end{equation}
\end{lem}
\begin{proof}
Let $\psi = (f,g)^T \in K_1^2$. 
Direct calculations yield
\begin{align*}
UL\psi &= U(d_A^\ast f + d_B^\ast g) = S(2d_A^\ast d_A - I)(d_A^\ast f + d_B^\ast g)\\
	&= S\{d_A^\ast f + (2d_A^\ast T - d_B^\ast) g\} = d_B^\ast f + (2d_B^\ast T - d_A^\ast) g
\end{align*}
On the other hand,
\begin{align*}
L\tilde T\psi &= L \begin{pmatrix}
0 & -I \\ I & 2T
\end{pmatrix}\begin{pmatrix}
f \\ g
\end{pmatrix}\\
&= L\begin{pmatrix}
-g \\ f+2Tg
\end{pmatrix}
= -d_A^\ast g + d_B^\ast (f+2Tg)
\end{align*}
and the proof is completed.
\end{proof}
Let $\mathcal{L}:= {\rm Im}L = d_A^\ast K_1+d_B^\ast K_1\subset K_2$, which is the center of our considerations in this paper.
First we have the following statement.
\begin{lem}
\label{lem-bijective}
The mapping $\tilde{T}$ is a bijective map on $K_1^2$.
\end{lem}
\begin{proof}
 If $ \tilde{T}\begin{pmatrix}f\\g\end{pmatrix}
=\begin{pmatrix}-g\\f+Tg\end{pmatrix} =0 $, we have
\begin{math}
  \begin{pmatrix}f\\g\end{pmatrix}=\begin{pmatrix}0\\0\end{pmatrix}
\end{math}.  Therefore, $\tilde{T}$ is injective.  On the other hand,
for any $\begin{pmatrix}f\\g\end{pmatrix}\in{K_2^2}$, we have that
\begin{displaymath}
  \tilde{T}\begin{pmatrix}g+2Tf\\-f\end{pmatrix}
  =
  \begin{pmatrix} 0 & -I \\ I & 2T \end{pmatrix}
  \begin{pmatrix}g+2Tf\\-f\end{pmatrix}
  =
  \begin{pmatrix}f\\g+2Tf-2Tf\end{pmatrix}
  =
  \begin{pmatrix}f\\g\end{pmatrix}.
\end{displaymath}
This implies that $\tilde{T}$ is surjective.
\end{proof}

Using this fact, we obtain the following.

\begin{lem}
\label{iden}
$U(\mathcal{L})=\mathcal{L}$.
\end{lem}
\begin{proof}
For any $\phi \in \mathcal{L}$, there is an element $\psi \in K_1^2$ such that $\phi = L\psi$.
Combining the statement of Lemma \ref{ULLT}, we have
\begin{equation*}
U\phi = UL\psi = L(\tilde T\psi)\in \mathcal{L},
\end{equation*}
which yields $U(\mathcal{L})\subset \mathcal{L}$.
\par
Conversely, for any $\psi \in K_1^2$, there is a unique element $\tilde \psi\in K_1^2$ such that $\tilde T \tilde \psi = \psi$, which follows from Lemma \ref{lem-bijective}; namely, $\tilde{T}$ is a bijection from $K_1^2$ onto itself.
Therefore, for any $\psi \in K_1^2$, we have
\begin{equation*}
L\psi = L\tilde T \tilde \psi = UL\tilde \psi \in U(\mathcal{L}),
\end{equation*}
which yields $\mathcal{L}\subset U(\mathcal{L})$ and the proof is completed.
\end{proof}
\begin{dfn}\rm
We say the invariant subspace $\mathcal{L}$ {\em the inherited eigenspace of $U$}.
The orthogonal complement $\mathcal{L}^\bot$ of $\mathcal{L}$ in $K$ is said {\em the birth eigenspace of $U$}.
\end{dfn}
Easy calculations yield that the subspace $\mathcal{L}^\bot$ is characterized by
\begin{equation}
\label{Lbot}
\mathcal{L}^\bot = \ker(d_A)\cap \ker(d_B).
\end{equation}

In \cite{HKSS}, ${\rm Spec}(U)$ is studied after decomposing it into two components: ${\rm Spec}(U|_{\mathcal{L}})$ and ${\rm Spec}(U|_{\mathcal{L}^\bot})$.
Note that this decomposition makes sense due to Lemma \ref{iden}.

\bigskip
Here we consider the eigenvalue problem on the inherited eigenspace $\mathcal{L}$ :
\begin{equation}
\label{target}
U|_{\mathcal{L}}\phi= \lambda \phi,\quad \phi \in \mathcal{L}.
\end{equation}

By Lemmas \ref{ULLT} and \ref{iden},
Eq.~(\ref{target}) is equivalent to the following problem: 
\par
\bigskip
{\em Find $\psi\notin \mathrm{ker}(L)$ and $\lambda\in \mathbb{C}$ such that }
\begin{equation}
\label{target'}
L(\lambda I-\tilde{T})\psi=0.
\end{equation}
By Eq.~(\ref{target'}), there are two possibilities:
\begin{description}
\item[(C1)'] $\psi\in \mathrm{ker}(\lambda I-\tilde{T})$ and $\psi\notin \mathrm{ker}(L)$.
\item[(C2)'] $\psi\notin \mathrm{ker}(\lambda I-\tilde{T})$, $\psi\in \mathrm{ker} (L(\lambda I-\tilde{T}))$ and $\psi\notin \mathrm{ker}(L)$.
\end{description}

Now a natural question arises: {\em What is $\ker L$ ?}
The following three lemmas answer this question and clarify our focus.
\par
First we have the following.
\begin{lem}
\label{kerL}
Let $L$ and $\tilde T$ be as above.
Then
\begin{equation}
\mathrm{ker}(L)=\mathrm{ker}(I-\tilde{T}^2). 
\end{equation}
\end{lem}
\begin{proof}
Let $\psi = (f,g)^T\in K_1^2$.
The statement $\psi \in \ker L$ implies $d_A^\ast f + d_B^\ast g = 0$.
Acting the operator $d_A$ on both sides, we have $f+Tg = 0$. 
Similarly, acting the operator $d_B$ on both sides, we also have $Tf+g = 0$. 
These observations imply
\begin{equation*}
f\in \ker (I-T^2),\quad g\in \ker (I-T^2)\quad \text{ with }\quad f = -Tg.
\end{equation*}
On the other hand,
\begin{equation*}
\tilde T^2 = \begin{pmatrix}
0 & -I\\ I & 2T
\end{pmatrix}
\begin{pmatrix}
0 & -I\\ I & 2T
\end{pmatrix}
= \begin{pmatrix}
-I & -2T\\ 2T & 4T^2-I
\end{pmatrix}.
\end{equation*}
Thus
\begin{equation}
\label{I-T^2}
(I-\tilde T^2)\begin{pmatrix}
f \\ g
\end{pmatrix}
= 
\begin{pmatrix}
-2I & -2T\\ 2T & 4T^2-2I
\end{pmatrix}
\begin{pmatrix}
f \\ g
\end{pmatrix}
= 
\begin{pmatrix}
-2f - 2Tg \\ 2Tf + 4T^2g -2g
\end{pmatrix}
=
\begin{pmatrix}
0 \\ 0
\end{pmatrix},
\end{equation}
which implies $\mathrm{ker}(L)\subset \mathrm{ker}(1-\tilde{T}^2)$.
\par
\bigskip
Conversely, assume $\psi = (f,g)^T\in \ker(I-\tilde T^2)$.
Then (\ref{I-T^2}) (in this case, it is a consequence of the assumption $\psi \in \ker(I-\tilde T^2)$) yields
$f +Tg = 0$ and $g + Tf = 0$.
Since $d_A d_A^\ast = I$ and $d_B d_B^\ast = I$, these two equations also imply
\begin{equation*}
d_A(d_A^\ast f + d_B^\ast g) = 0\quad \text{ and }\quad d_B(d_A^\ast f + d_B^\ast g) = 0.
\end{equation*}
Thus $d_A^\ast f + d_B^\ast g \in \ker d_A\cap \ker d_B = \mathcal{L}^\ast$ and hence $d_A^\ast f + d_B^\ast g \in \mathcal{L}\cap \mathcal{L}^\bot = \{0\}$.
Finally we have $d_A^\ast f + d_B^\ast g = 0$ and hence $\ker(I-\tilde T^2) \subset \ker L$.
\end{proof}
Lemma \ref{kerL} indicates that Eq.~(\ref{target'}) is equivalent to 
\begin{equation}
\label{target''}
(I-\tilde{T}^2)(\lambda I -\tilde{T})\psi=0\quad \text{ and } \quad \psi \not \in \ker(I-\tilde T^2).
\end{equation}
Consequently, the case (C2)' is equivalent to the following:
\begin{description}
\item[(C2)''] $\psi \not \in \ker(\lambda I-\tilde T)$, $\psi \not \in \ker(I-\tilde T^2)$ and $(I-\tilde T^2)(\lambda I-\tilde T)\psi = 0$.
\end{description}

We thus have translated the structure of $\ker L$ into corresponding nullspaces of $\tilde T$.
Thanks to this fact, we can provide the following lemmas.

\begin{lem}
\label{idenspec1}
Let $\tilde T$, $U$ and $\mathcal{L}$ be as above.
Then we have
\begin{equation}
{\rm Spec}(\tilde{T})\supset {\rm Spec}(U|_{\mathcal{L}}).
\end{equation}  
\end{lem}
\begin{proof}
Assume that $\lambda \in {\rm Spec}(U|_{\mathcal{L}})$.
If Eq.~(\ref{target''}) holds with $\lambda\in \rho(\tilde{T}) := \mathbb{C}\setminus {\rm Spec}(\tilde T)$, then the operator $\lambda I-\tilde T$ has bounded inverse and $\psi=(\lambda I-\tilde{T})^{-1}\phi$, where $\phi\in \mathrm{ker}(I-\tilde{T}^2)$. 
Since $(\lambda I-\tilde{T})^{-1}$ and $(I-\tilde{T}^2)$ are commutative, $L\psi=(\lambda I-\tilde{T})^{-1}(I-\tilde{T}^2)\phi=0$, which contradicts
$\psi \not \in \ker L$.
Hence $\lambda \in {\rm Spec}(\tilde T)$ and the proof is completed.
\end{proof}
The spectra of $\tilde T$ and $U|_{\mathcal{L}}$ are coincide except $\pm1$, as shown in the following lemma.
\begin{lem}
\label{lem-iden-not-pm1}
Let $\tilde T$, $U$ and $\mathcal{L}$ be as above.
Then we have
\begin{equation*}
{\rm Spec}(\tilde{T})\setminus\{\pm 1 \}={\rm Spec}(U|_{\mathcal{L}})\setminus\{\pm 1 \}.
\end{equation*}
\end{lem} 
\begin{proof}
By Lemma~\ref{idenspec1}, ${\rm Spec}(\tilde{T})\setminus\{\pm 1 \} \supset {\rm Spec}(U|_{\mathcal{L}})\setminus\{\pm 1 \}$. 
So we need to show that ${\rm Spec}(\tilde{T})\setminus\{\pm 1 \} \subset {\rm Spec}(U|_{\mathcal{L}})\setminus\{\pm 1 \}$. 
If $\lambda \not = \pm 1$ holds in (\ref{target''}), we have $\ker (I-\tilde{T}^2)(\lambda I-\tilde{T}) \equiv \ker(I-\tilde{T}^2)\oplus \ker(\lambda I-\tilde{T})$. 
Since $\psi \not\in \ker(I-\tilde T^2)$, which is the second condition of (\ref{target''}),
we have $\psi \in \ker (\lambda I - \tilde T)$.
This is exactly the case (C1)' and hence $\lambda \in {\rm Spec}(U|_{\mathcal{L}})$. 
\end{proof}

\bigskip
We have seen that, from Lemma \ref{kerL}, (\ref{target'}) is equivalent to (\ref{target''}).
Using this fact and Lemma \ref{lem-iden-not-pm1}, we have the following equivalences.
\begin{prop}
\label{prop-equiv}
(C1)' is equivalent to
\begin{description}
\item[(C1)] $\lambda \in {\rm Spec}(\tilde T)\setminus \{\pm 1\}$ and (\ref{target''}). 
\end{description}
Similarly, (C2)', namely (C2)'', is equivalent to
\begin{description}
\item[(C2)] $\lambda \in {\rm Spec}(\tilde T)\cap \{\pm 1\}$ and (\ref{target''}). 
\end{description}
\end{prop}
\begin{proof}
Our target problem (\ref{target''}) is divided into two disjoint cases (C1) and (C2), 
and also other two disjoint cases (C1)' and (C2)''. 
It is therefore sufficient to show the equivalence between 
one of them, i.e., \lq\lq { }(C1) and (C1)' " or \lq\lq { }(C2) and (C2)" ", to show both equivalence.
\par
In what follows consider the equivalence between (C1)' and (C1).
Proof of Lemma \ref{lem-iden-not-pm1} indicates that (C1) implies (C1)'.
It is thus sufficient to consider whether (C1)' implies (C1).
If not, $\lambda \in \{\pm 1\}$ may also satisfy (C1)'.
For example, assume that $\lambda=1$ satisfies (C1)'.
We then have
\begin{equation*}
\psi \in \ker(I-\tilde T)\subset \ker(I+\tilde T)(I-\tilde T) = \ker(I-\tilde T^2),
\end{equation*}
which contradicts (\ref{target''}).
Similar arguments holds for $\lambda = -1$.
We thus obtain (C1)' is equivalent to (C1) and complete the proof.
\end{proof}

Proposition \ref{prop-equiv} guarantees that the study of ${\rm Spec}(U|_{\mathcal{L}})$ is reduced to individual cases (C1) and (C2).

\subsection{The case (C1)}
\label{section-first}

The problem in the setting of case (C1) is reduced to the one that we
find $\lambda\in \mathbb{C}$ with $\lambda \neq \pm 1$ and $\psi\neq 0$ such that $(\lambda I -\tilde{T})\psi=0$.
We then have 
	\begin{align*}
        (\lambda I -\tilde{T})\psi=0 &\Leftrightarrow \begin{pmatrix} \lambda I & I \\ -I & \lambda I-2T \end{pmatrix} \begin{pmatrix} f \\ g \end{pmatrix}= \begin{pmatrix} 0 \\ 0 \end{pmatrix} \\
        	&\Leftrightarrow \lambda f + g = 0,\quad -f + (\lambda I - 2T)g = 0 \\
                &\Leftrightarrow g\in \ker(\lambda^2 I -2\lambda T+I) \;\mathrm{and}\; f=(\lambda I -2T)g.
        \end{align*}
Since $\tilde{T}$ is invertible from Lemma \ref{lem-bijective}, $0\notin \mathrm{Spec}(\tilde{T})$ holds true, which means $\lambda\neq 0$. 
The above statement is thus equivalent to 
	\[ g\in\mathrm{ker}\left(\left(\frac{\lambda+\lambda^{-1}}{2}\right) I-T\right)\;\mathrm{and}\;f=-\lambda^{-1}g. \]
Therefore we have 
	\[ \mathrm{ker}(\lambda-\tilde{T})=\begin{pmatrix} 1 \\ -\lambda \end{pmatrix}\otimes \mathrm{ker}\left(\left( \frac{\lambda+\lambda^{-1}}{2}\right) I -T\right). \]
Putting $\mu=(\lambda+\lambda^{-1})/2$, we have 
	\[ \mathrm{ker}(e^{\pm i{\arccos \mu}}I -\tilde{T})=\begin{pmatrix} 1 \\ -e^{\pm i{\arccos \mu}} \end{pmatrix}\otimes \mathrm{ker}\left(\mu I-T\right). \]
Remarking $e^{\pm i{\arccos \mu}}=\pm 1$ if and only if $\mu=\pm 1$, we have the following lemma, which describes eigenpairs of $U$ associated with $\lambda \not = \pm 1$.

\begin{lem}
\label{lem-C1}
The eigenpair of $U$ associated with $\lambda \in \mathrm{Spec}(U|_\mathcal{L})\setminus \{\pm 1\}$ is characterized by the following.
\begin{align*}
        \mathrm{Spec}(U|_\mathcal{L})\setminus \{\pm 1\} &= \{ e^{\pm i{\arccos \mu}}| \mu\in \mathrm{Spec}(T)\setminus\{\pm 1\} \}, \\
        \mathrm{ker}(e^{\pm i{\arccos \mu}} I -U|_\mathcal{L}) &= (d_A^\ast - e^{\pm i{\arccos \mu}} d_B^\ast )\mathrm{ker}\left(\mu I-T\right).
\end{align*}
\end{lem}

\subsection{The case (C2)}
\noindent
Our main aim here is the complete description of ${\rm Spec}(U|_\mathcal{L})$ in the case (C2).
The key point is the structure of the eigenspaces $\ker (I\pm \tilde T)$ as well as the {\em generalized eigenspaces} $\ker (I\pm \tilde T)^n$, $n\geq 2$.
To this end, we provide the following lemma.
\begin{lem}
\label{pow-I-tilde-T}
For each $n\geq 1$, we have
\begin{align}
\label{pow+1}
(I-\tilde{T})^{2n} &= 2^n (-\tilde T)^n \begin{pmatrix}
(I-T)^n & 0\\
0 & (I-T)^n 
\end{pmatrix},\\
\label{pow-1}
(I+\tilde{T})^{2n} &= 2^n \tilde T^n \begin{pmatrix}
(I+T)^n & 0\\
0 & (I+T)^n 
\end{pmatrix}.
\end{align}
\end{lem}

\begin{proof}
First consider $(I-\tilde{T})^{2n}$.
The case $n=0$ is trivial.
We have
\begin{align*}
(I-\tilde{T})^2 &= \begin{pmatrix}
I & I\\
-I & I-2T
\end{pmatrix}\begin{pmatrix}
I & I\\
-I & I-2T
\end{pmatrix}\\
&=\begin{pmatrix}
0 & 2(I-T)\\
-2(I-T) & -4T(I-T)
\end{pmatrix}
=2\begin{pmatrix}
0 & I\\
-I & -2T
\end{pmatrix}\begin{pmatrix}
I-T & 0\\
0 & I-T
\end{pmatrix},
\end{align*}
which means (\ref{pow+1}) for $n=1$.
Notice that
\begin{equation}
\label{comm+1}
\begin{pmatrix}
I-T & 0\\
0 & I-T
\end{pmatrix}
\begin{pmatrix}
0 & I\\
-I & -2T
\end{pmatrix}
=
\begin{pmatrix}
0 & I\\
-I & -2T
\end{pmatrix}\begin{pmatrix}
I-T & 0\\
0 & I-T
\end{pmatrix}.
\end{equation}
Assume that (\ref{pow+1}) holds for some $n=n_0 \geq 1$. Then
\begin{align*}
(I-\tilde{T})^{2(n_0+1)} &= 
2^{n_0} (-\tilde T)^{n_0} \begin{pmatrix}
(I-T)^{n_0} & 0\\
0 & (I-T)^{n_0} 
\end{pmatrix}\begin{pmatrix}
0 & 2(I-T)\\
-2(I-T) & -4T(I-T)
\end{pmatrix}\\
&=2^{n_0+1} (-\tilde T)^{n_0} \begin{pmatrix}
(I-T)^{n_0} & 0\\
0 & (I-T)^{n_0} 
\end{pmatrix} (-\tilde T)\begin{pmatrix}
I-T & 0\\
0 & I-T
\end{pmatrix}\\
&=2^{n_0+1} (-\tilde T)^{n_0+1} \begin{pmatrix}
(I-T)^{n_0+1} & 0\\
0 & (I-T)^{n_0+1}
\end{pmatrix}
\end{align*}
by (\ref{comm+1}), which proves (\ref{pow+1}) for $n=n_0+1$.
By induction, (\ref{pow+1}) holds for all $n\geq 0$.
\par
\bigskip
Next consider $(I+\tilde{T})^{2n}$.
The case $n=0$ is trivial.
We have
\begin{align*}
(I+\tilde{T})^2 &= \begin{pmatrix}
I & -I\\
I & I+2T
\end{pmatrix}\begin{pmatrix}
I & -I\\
I & I+2T
\end{pmatrix}\\
&=\begin{pmatrix}
0 & -2(I+T)\\
2(I+T) & 4T(I+T)
\end{pmatrix}
=2\begin{pmatrix}
0 & -I\\
I & 2T
\end{pmatrix}\begin{pmatrix}
I+T & 0\\
0 & I+T
\end{pmatrix},
\end{align*}
which means (\ref{pow-1}) for $n=1$.
By the same arguments as the proof of (\ref{pow+1}) we obtain (\ref{pow-1}).
\end{proof}

With the help of Lemma \ref{pow-I-tilde-T}, we can prove the following, which gives us the description of eigenspaces of $\tilde T$ associated with eigenvalues $\lambda = \pm 1$.
\begin{prop}
\label{prop-gen-ker}
	\begin{align} 
	\label{ker+1}
        \mathrm{ker}(I-\tilde{T}) &= \begin{pmatrix}
1 \\ -1
\end{pmatrix} \otimes \ker(I-T),\\
	\label{gker+1}
        \mathrm{ker}(I-\tilde{T})^n &= \mathbb{C}^2 \otimes \mathrm{ker}(I-T) \;\;(n\geq 2),\\
	\label{ker-1}
        \mathrm{ker}(I+\tilde{T}) &=  \begin{pmatrix}
1 \\ 1
\end{pmatrix} \otimes \ker(I+T), \\
	\label{gker-1}
        \mathrm{ker}(I+\tilde{T})^n &= \mathbb{C}^2 \otimes \mathrm{ker}(I+T) \;\;(n\geq 2).
        \end{align}
\end{prop}
\begin{proof}
We only prove (\ref{ker+1}) and (\ref{gker+1}).
Remaining statements (\ref{ker-1}) and (\ref{gker-1}) can be proved by the same arguments.
\par
First we have
\begin{align}
\label{pow-I-T}
I-\tilde T &= \begin{pmatrix}
I & I\\
-I & I-2T
\end{pmatrix},\quad 
(I-\tilde T)^2 = \begin{pmatrix}
0 & 2(I-T)\\
-2(I-T) & -4T(I-T)
\end{pmatrix}.
\end{align}
For $\psi = (f,g)^T\in \ker (I-\tilde T)$, we have
\begin{equation*}
\begin{pmatrix}
I & I\\
-I & I-2T
\end{pmatrix}
\begin{pmatrix}
f \\ g
\end{pmatrix}
=
\begin{pmatrix}
f+g \\ -f + (I-2T)g
\end{pmatrix}
=
\begin{pmatrix}
0 \\ 0
\end{pmatrix},
\end{equation*}
which yields $f=-g$ and $Tf=f$, and hence $\mathrm{ker}(I-\tilde{T}) \subset \left\{ (f,-f)^T: f\in \mathrm{ker}(I-T) \right\}$.
The converse is trivial and (\ref{ker+1}) holds true.
\par
Secondly, consider $(I-\tilde T)^2 \psi = 0$.
(\ref{pow-I-T}) immediately yields
\begin{align*}
(I-\tilde T)^2 \psi &= \begin{pmatrix}
0 & 2(I-T)\\
-2(I-T) & -4T(I-T)
\end{pmatrix}
\begin{pmatrix}
f \\ g
\end{pmatrix}
=
\begin{pmatrix}
0 \\ 0
\end{pmatrix}\\
&\Leftrightarrow
\quad f\in \ker(I-T),\quad g\in \ker(I-T).
\end{align*}
Next consider $(I-\tilde T)^4 \psi = 0$.
By Lemma \ref{pow-I-tilde-T} with $n=2$, we have
\begin{align*}
(I-\tilde T)^4 \psi &= 2^2(-\tilde T)^2\begin{pmatrix}
(I-T)^2 & 0\\
0 & (I-T)^2
\end{pmatrix}
\begin{pmatrix}
f \\ g
\end{pmatrix}
=
\begin{pmatrix}
0 \\ 0
\end{pmatrix}.
\end{align*}
Since the operator $-\tilde T$ is invertible by Lemma \ref{lem-bijective}, the above equation implies $f\in \ker(I-T)^2$ and $g\in \ker (I-T)^2$.
\par
Here note that $\ker (I-T)^2 = \ker (I-T)$, since $T$ is Hermitian and hence diagonalizable and it implies $\ker (I-T)^2 = \ker (I-T)$.
Thus we observe that $f\in \ker (I-T)$ and $g\in \ker (I-T)$.
In particular, $\ker (I-\tilde T)^4 = \ker (I-\tilde T)^2$ holds.
The similar arguments holds for all $n$ by Lemma \ref{pow-I-tilde-T} and the invertibility of $-\tilde T$.
\par
In general, $\ker (I-\tilde T)^n \subset \ker (I-\tilde T)^{n+1}$ holds for all $n$, which is a fundamental property from linear algebra.
Combining the fact $\ker (I-\tilde T)^4 = \ker (I-\tilde T)^2$, we have $\ker (I-\tilde T)^2 \subset \ker (I-\tilde T)^3 \subset \ker (I-\tilde T)^4 = \ker (I-\tilde T)^2$.
\par
Finally we have $\ker (I-\tilde T)^n = \ker (I-\tilde T)^2$ for all $n\geq 2$ by recursive arguments and the proof is completed.
\end{proof}

\begin{rem}\rm
In the proof of Proposition \ref{prop-gen-ker}, the operator $T$ being self-adjoint is used to guarantee $\ker(I-T)^2 = \ker(I-T)$.
Our arguments also hold even for non-Hermitian operators as long as the algebraic multiplicity of every eigenvalues coincides with their geometric multiplicity.
\end{rem}

We provide a lemma from linear algebra before stating the next proposition.
\begin{lem}
\label{lem-directsum}
Let $U, V$ and $W$ be vector spaces such that $U\subset V$ and that $V\cap W=\{0\}$.
Then the relationship $\phi\in (U\oplus W)\cap V$ implies $\phi\in U$.
\end{lem}
\begin{proof}
Let $\phi\in (U\oplus W)\cap V$. 
Then there are unique elements $f\in U$ and $g\in W$ such that $\phi = f+g$.
On the other hand, we also have $\phi\in V\oplus W$ since $U\subset V$ and there are also unique elements $f'\in V$ and $g'\in W$ such that $\phi = f'+g'$.
Since
\begin{equation*}
0 = \phi - \phi = (f-f') + (g-g'),
\end{equation*}
we have $V\ni f-f' = -(g-g')\in W$. By the relationship $V\cap W=\{0\}$, $f=f'$ and $g=g'$ hold.
Now we assumed $\phi\in V$ (by $\phi\in (U\oplus W)\cap V$), which implies $\phi = f' \in V$ and $g' = 0$ by the uniqueness of decompositions in $V\oplus W$.
Finally we obtain $\phi = f' = f \in U$.
\end{proof}

The following proposition characterizes the eigenvalues $\pm 1$ of $U|_{\mathcal{L}}$.

\begin{prop}
\label{prop-pm1}
For $\psi = (f,g)^T\in K_1^2$, the following four statements are equivalent.
\begin{description}
\item[{\rm (i)}] $UL\psi = \pm L\psi \not = 0$.
\item[{\rm (ii)}] $\psi \in Z\setminus\ker(L)$ where $Z=\ker(I-\tilde T)^2 (I+\tilde T).$
\item[{\rm (iii)}] $\psi =\psi_1 + \psi_2$, such that $\psi_1 \in X_\mp \equiv \ker(I\mp \tilde T)^2 \setminus \ker(I\mp \tilde T)$ with $\psi_1\not = 0$ and that $\psi_2 \in \ker(L)$.
\item[{\rm (iv)}] $f=f_1+f_2$ and $g=g_1+g_2$, such that $f_1,g_1\in \ker(I\mp T)$ with $f_1\pm g_1\not = 0$ and that $(f_2,g_2)^T\in \ker(L)$.
\end{description}
\end{prop}
\begin{proof}
First notice that since $I-\tilde T^2=(I+\tilde T)(I-\tilde T)$, 
if $\tilde{T}\psi=\pm\psi\ne0$, 
then $L\psi=0$ by Lemma~\ref{kerL}, 
and thus $L\psi$ is {\em not} an eigenfunction of $U|_{\mathcal{L}}$.
\par
In what follows we consider $\lambda = +1\in {\rm Spec}(U|_\mathcal{L})$.
The similar arguments to below yield the corresponding equivalence for the eigenvalue $\lambda = -1$.
\par
Since $+1\in {\rm Spec}(U|_\mathcal{L})$ if and only if 
there exists $\phi\in\mathcal{L}$ such that $U\phi=\phi\ne0$, we
can see that 
$1\in {\rm Spec}(U|_\mathcal{L})$ is equivalent to that
there exists $\psi \in K_1^2$ such that
\begin{equation*}
  L\psi \not = 0\quad \text{ and }\quad L(I-\tilde T)\psi =0,
\end{equation*}
which is equivalent to 
\begin{equation*}
  L\psi \not = 0\quad \text{ and }\quad (I+\tilde T)(I-\tilde T)^2 \psi=0
\end{equation*}
by Lemma \ref{kerL}. 
This is also equivalent to
\begin{equation*}
\psi \not \in \ker(L) \quad \text{ and }\quad \psi \in Z=\ker(I-\tilde T)^2 (I+\tilde T).
\end{equation*}
Thus we obtain that (i) and (ii) are equivalent to each other.
Note that $Z$ is a vector space.
\par
Now we assume $\psi \in Y\equiv Z\setminus \ker(L)$.
Then there are two vectors $\psi_1\in \ker(I-\tilde T)^2$ and $\psi_2 \in \ker(I+\tilde T)$ such that $\psi = \psi_1 + \psi_2 \not \in \ker(L)$.
Now it holds that $\psi_2 \in \ker(I+\tilde T) \subset \ker(L)$. 
If $\psi_1$ belongs to $\ker(I-\tilde T)$, then $\psi_1\in \ker(L)$ also holds and hence $\psi = \psi_1 + \psi_2 \in \ker(L)$, which is contradiction.
Therefore $\psi_1\in \ker(I-\tilde T)^2\setminus \ker(I-\tilde T) = X_-$.
This shows \lq\lq (ii) $\Rightarrow$ (iii)".
\par
\bigskip
Conversely, we assume (iii), namely, $\psi= \psi_1'+\psi_2'$ with $0\not =\psi_1' \in X_-$ and $\psi_2'\in \ker(L)$.
First we prove $\psi\in Z$. 
Since $\ker(I-\tilde T)\subset \ker(I-\tilde T)^2$, we have $\ker(L) = \ker(I-\tilde T)\oplus \ker(I+\tilde T)\subset Z$, which implies $\psi_2'\in Z$.
We also have $\psi_1'\in X_- \subset \ker(I-\tilde T)^2\subset Z$ and hence $\psi = \psi_1'+\psi_2' \in Z$.
Next we prove $\psi\not \in \ker(L)$.
By definition, we have $\psi_2'\in \ker(L)$ and hence our claim is reduced to the statement $(0\not =)\psi_1'\not \in \ker(L)$.
Assume that it is not the case, namely, $\psi_1'\in \ker(L)$.
Since $X_-\subset \ker(I-\tilde T)^2$ and 
\begin{equation*}
\ker(I-\tilde T)^2\cap \ker(I+\tilde T) = \{0\}
\end{equation*}
holds, we have $(0\not =) \psi_1' \in \ker(I- \tilde T)$, which follows from Lemma \ref{lem-directsum} with $U=\ker(I-\tilde T), V= \ker(I-\tilde T)^2$ and $W=\ker(I+\tilde T)$, and it contradicts the assumption $\psi'_1\in X_-$. 
It thus holds that $\psi_1'\not \in \ker(L)$.
Finally we have $\psi = \psi_1' + \psi_2' \not \in \ker(L)$, since $\ker(L)$ is a vector space, and hence $\psi\in Z\setminus \ker(L) = Y$.
\par
Consequently, (ii) and (iii) are equivalent.
\par
\bigskip
Next assume $\psi=\psi_1+\psi_2$ with $0\not = \psi_1\in X_-$ and $\psi_2\in \ker(L)$.
Write $\psi_i = (f_i, g_i)^T\in K_1^2$ for $i=1,2$.
By Proposition \ref{prop-gen-ker}, $\psi_1 = (f_1,g_1)^T$ with $f_1 + g_1 \not = 0$ and $f_1,g_1 \in \ker(I-T)$ hold true.
This implies \lq\lq (iii) $\Rightarrow$ (iv)". 
The converse also follows from Proposition \ref{prop-gen-ker}.
\end{proof}

This proposition indicates that the eigenfunctions of $U|_{\mathcal{L}}$ associated with the eigenvalues $\lambda = \pm 1$ are characterized by purely generalized kernels $X_\mp = \ker(I\mp \tilde T)^2 \setminus \ker(I\mp \tilde T)$ up to $\ker(L)$.

\subsection{The final result}
\label{section-final}
Summarizing the above arguments, we have the following spectral mapping theorem of $U$.

\begin{thm}[Spectral Mapping Theorem of $U$]
\label{thm-spec}
Let $\varphi_{QW}(x)=(x+x^{-1})/2$ be the Joukowsky transform. Then
\begin{equation}
\label{spec-U}
\ker (\lambda I - U|_\mathcal{L}) =
                \begin{cases}
                 \; L(\ker(\lambda I-\tilde{T})) & \text{: $\lambda\notin \{\pm 1\}$}, \\
                  \{L\psi\mid \psi = (f,g)^T, f,g\in \ker(I-T), f+g\not = 0\} &\\
                  \quad= L(\ker(I-\tilde T)^2 \setminus \ker(I-\tilde T)) &  \text{: $\lambda=1$},\\
                  \{L\psi\mid \psi = (f,g)^T, f,g\in \ker(I+T), f-g\not = 0\} &\\
                  \quad= L(\ker(I+\tilde T)^2 \setminus \ker(I+\tilde T)) &  \text{: $\lambda=-1$}.
                \end{cases}
\end{equation}
In particular, we have
\begin{equation}
{\rm Spec}(U|_\mathcal{L})={\rm Spec}(\tilde{T})=\varphi_{QW}^{-1}({\rm Spec}(T)).
\end{equation}
For $\lambda \not = \pm 1$, the eigenspace $L(\ker(\lambda I-\tilde{T}))$ is given by
\begin{equation*}
\left\{d_A^\ast f_{\varphi_{QW}(\lambda)}-\lambda d_B^\ast f_{\varphi_{QW}(\lambda)} \mid  f_{\varphi_{QW}(\lambda)}\in \ker(\varphi_{QW}(\lambda) I -T)\right\}.
\end{equation*}
The spectrum of $U|_{\mathcal{L}^\bot}$ has the following relationship : ${\rm Spec}(U|_{\mathcal{L}^\bot}) = {\rm Spec}(-S)$.
Moreover, we have
\begin{equation}
\label{spec-Ubot}
\ker (\lambda I - U|_{\mathcal{L}^\bot}) =
                \begin{cases}
                 \; 0 & \text{: $\lambda\not = \pm 1$},\\
                  \ker (d_A)\cap \ker(I+S) &  \text{: $\lambda=1$},\\
                  \ker (d_A)\cap \ker(I-S) &  \text{: $\lambda=-1$}.
                \end{cases}
\end{equation}
\end{thm}

\begin{proof}
Eigenstructures of $U|_{\mathcal{L}}$ with $\lambda = \pm 1$ directly follow from Proposition \ref{prop-pm1}.
Note that $\varphi_{QW} (\pm 1) = \pm 1$, which yield $\ker (\varphi_{QW}(\pm 1)I-T) = \ker (I\mp T)$.
For $\lambda \not = \pm 1$, Lemma \ref{lem-iden-not-pm1} says ${\rm Spec}(U|_\mathcal{L})\setminus \{\pm 1\} = {\rm Spec}(\tilde{T})\setminus \{\pm 1\}$.
The spectral mapping property ${\rm Spec}(\tilde{T})\setminus \{\pm 1\} = \varphi_{QW}^{-1}({\rm Spec}(\tilde{T})\setminus \{\pm 1\}) = \varphi_{QW}^{-1}({\rm Spec}(\tilde{T}))\setminus \{\pm 1\}$ is the consequence of Lemma \ref{lem-C1}.

\par
Let $\phi \in \mathcal{L}^\bot$ be an eigenfunction of $U$. 
Then $U\phi = S(2d_A^\ast d_A - I)\psi = -S\phi$ and hence $\phi$ is an eigenfunction of $-S$.
Note that ${\rm Spec}(S) =\{\pm 1\}$, since $S$ is an involution, namely, self-adjoint and unitary.
In such a case, we also have $d_B \phi = d_A S\phi = \mp d_A \phi = 0$.
Thus the statement $\phi \in \mathcal{L}^\bot \cap \ker(I\pm S)$ is consequently equivalent to $\phi\in \ker(d_A) \cap \ker(I\pm S)$.
\par
Combination of these results yields our statement.
\end{proof}

Before stating the corollary of Theorem \ref{thm-spec}, we provide the following lemma.
\begin{lem}
\label{lem-ker-dAdB}
For $f\in K_1$, $f\in \ker(I\mp T)$ holds if and only if $d_A^\ast f = \pm d_B^\ast f$.
\end{lem}

\begin{proof}
We immediately have
\begin{align*}
f\in \mathrm{ker}(I-T) &\Leftrightarrow d_Bd_A^\ast  f=f \Leftrightarrow d_B(d_A^\ast f-d_B^\ast f)=d_A(d_A^\ast f-d_B^\ast f)=0 \\
& \Leftrightarrow d_A^\ast f-d_B^\ast f\in \mathcal{L}\cap (\ker d_A \cap \ker d_B)=\mathcal{L}\cap\mathcal{L}^\bot = \{0\}\\ 
& \Leftrightarrow d_A^\ast f=d_B^\ast f. 
\end{align*}
The same calculations yield $f\in \ker(I+T) \Leftrightarrow d_A^\ast f = -d_B^\ast f$.
\end{proof}

Summarizing the above statements with the preceding work \cite{HKSS}, we obtain the complete characterization of the spectrum of quantum walks on graphs.

\begin{cor}[\cite{HKSS}, \cite{HS2004}, Complete eigenstructure of $U$ on the graph $G$]
\label{cor-HKSS}
Assume that a connected graph $G$ is finite; namely, $|V(G)|$, $|E(G)| < \infty$.
Consider the Szegedy walk on $G$ given in Example \ref{ex-Szegedy}.
Then we have the following statements.
\begin{enumerate}
\item Integers $m_{\pm 1}\in \{0,1\}$ denote the multiplicity of eigenvalues $\pm 1$ of $T$, respectively, which are given by
\begin{equation*}
m_+ = \begin{cases}
	 1 & \text{$1\in {\rm Spec}(T)$}\\
	 0 & \text{otherwise}
\end{cases},\quad 
m_- = \begin{cases}
	 1 & \text{$1\in {\rm Spec}(T)$ and $G$ is bipartite}\\
	 0 & \text{otherwise}
\end{cases}.
\end{equation*}
See \cite{HKSS} for details of these definitions.
\par
Let $\varphi_{QW} : \mathbb{C}\to \mathbb{C}$ be the Joukowsky transform given by $\varphi_{QW}(x) = (x+x^{-1})/2$.
Then we have
\begin{equation*}
{\rm Spec}(U) = \varphi_{QW}^{-1}({\rm Spec}(T)) \cup \{+1\}^{M_{+1}} \cup \{-1\}^{M_{-1}},
\end{equation*}
where $M_{+1}=\max \{0, |E|-|V|+m_{+1}\}$, $M_{-1}=\max \{0, |E|-|V|+m_{-1}\}$.
\item The eigenfunctions of the eigenvalue $\lambda \in \varphi_{QW}^{-1}({\rm Spec}(T))$ generating $\mathcal{L}$ are given by 
\begin{equation*}
\begin{cases}
\frac{1}{\sqrt{2}|\sin \lambda|}(I-e^{i\lambda}S) d_A^\ast f, \text{ where }f\in \ker(\varphi_{QW}(\lambda)I-T)  & :\ \lambda \not = \pm 1,\\
d_A^\ast f, \text{ where } f\in \ker(I-T) & \text{: $\lambda=1$}.
\end{cases}
\end{equation*}
If $G$ is further assumed to be bipartite, then $U|_{\mathcal{L}}$ has the eigenvalue $-1$ and its eigenfunction is given by
\begin{equation*}
d_A^\ast f, \text{ where }f\in \ker(I+T).
\end{equation*}

The eigenspaces corresponding to eigenvalues $\{1\}^{M_{+1}}$ and $\{-1\}^{M_{-1}}$ except those corresponding to $\varphi_{QW}^{-1}(\pm 1)$ are described by 
\begin{equation}
\label{eigenG-1}
\ker (d_A)\cap \mathcal{H}^{S}_{-1} \quad \text{ and }\quad \ker (d_A)\cap \mathcal{H}^{S}_{+1},
\end{equation}
respectively. Here $\mathcal{H}^{S}_{\pm 1} = \ker (I\mp S)$.
\end{enumerate}
\end{cor}

\begin{proof}
By definition of $m_+$, we have $\dim \ker(I-T)=1$.
Thus $\psi = (f,g)^T$ in (\ref{spec-U}) is written as $(f,\alpha f)^T$ for $\alpha \in \mathbb{R}\setminus \{-1\}$.
By Lemma \ref{lem-ker-dAdB}, we have $d_A^\ast f = d_B^\ast f$.
It thus follows that the eigenfunction $L\psi$ of $U|_{\mathcal{L}}$ is written as $\beta d_A^\ast f$ for some $\beta \in \mathbb{R}\setminus \{0\}$.
\par
Properties that $G$ is bipartite, we have $\dim \ker (I+T) = 1$ by definition of $m_-$ (cf. \cite{HKSS}).
The same arguments as above thus yield (\ref{eigenG-1}).
\par
All remaining statements follow from Theorem \ref{thm-spec}.
\end{proof}

\begin{rem}\rm
It is shown in \cite{HKSS} that the discriminant operator $T$ has the eigenvalue $1$ {\em if and only if the underlying random walk on $G$ has a reversible measure}.
The operator $T$ has the eigenvalue $-1$ if the graph $G$ is bipartite, 
namely, its vertex-set can be partitioned into two parts $V_1$ and $V_2$ such that each edge has one vertex in $V_1$ and one vertex in $V_2$.
Spectral properties around $1$ of the twisted random walks has been well studied in \cite{HS2004} and also its induced quantum walk has been studied in \cite{HKSS}.
Analysis in \cite{HKSS} also indicates that the eigenstructure of ${\rm Spec}(U|_{\mathcal{L}^\bot})$ induces localization of $U$.
Roughly, a cycle structure on $G$ induces localization, 
since the geometric multiplicities of $\{\pm 1\}\cap U|_{\mathcal{L}^\bot}$ are described with the first Betti number of $G$, which becomes a motivation of extend quantum walks on graphs to new quantum walks on simplicial complexes \cite{MOS}. 
See \cite{HKSS} for more detailed description of $\ker (I\pm U|_{\mathcal{L}^\bot})$.
\end{rem}

\section{Conclusion}
In this paper, we have discussed the spectral mapping theorem of quantum walks of the general form containing well-known Grover and Szegedy walks on finite graphs.
In addition to known characterizations of eigenstructure for $\lambda \in \left({\rm Spec}(U|_{\mathcal{L}})\cap \{r\in \mathbb{C}\mid {\rm Re}r \in (-1,1)\}\right)$ and $\lambda\in {\rm Spec}(U|_{\mathcal{L}^\bot})$ in preceding works, we derived a more detailed description of $\pm 1$ as eigenvalues of $U|_\mathcal{L}$.
In particular, we have derived the spectral mapping theorem of $U|_{\mathcal{L}}$ without using the spectral decomposition of the discriminant operator $T$ to obtain the description of eigenstructures.
\par
We have seen all of the eigenvalues of $\tilde{T}$ describe whole eigenvalues of $U|_\mathcal{L}$ even if $\tilde{T}$ is not diagonalizable. 
In the case of $\lambda\in\mathrm{Spec}(U|_\mathcal{L})\setminus\{\pm 1\}$, $\mathrm{ker}(\lambda I-U|_\mathcal{L})$ is obtained by $L(\mathrm{ker}(\lambda I-\tilde{T}))$, 
on the other hand, in the case of the eigenvalue where the geometric multiplicities degenerate,  $\mathrm{ker}(\lambda I -U|_\mathcal{L})$ is described by purely generalized eigenspace of $\lambda I-\tilde{T}$. 
This observation is a by-product of introducing our new method. 
\par
Applicability of our approach for describing ${\rm Spec}(U)$ in infinite dimensional setting (e.g. \cite{O2016, SS2015-1, SS2015-2}) remains open.
The difficulty comes from the essential spectrum of $T$.
Suitable assumptions for quantum walk models and careful treatments of ${\rm Spec}(T)$ will yield the extension of our results in infinite dimensional setting, which will be our future work.

\section*{Acknowledgements}
KM was partially supported by Coop with Math Program, a commissioned project by MEXT. OO was partially supported by JSPS KAKENHI Grant Number 24540208. ES was was partially supported by JSPS Grant-in-Aid for Young Scientists (B) (No. 25800088) and Japan-Korea Basic Scientific Cooperation Program \lq\lq Non-commutative Stochastic Analysis; New Aspects of Quantum White Noise and Quantum Walks" (2015--2016).

\bibliographystyle{plain}
\bibliography{proceeding_YNU}

\end{document}